\documentclass[pdflatex,sn-mathphys-num]{sn-jnl}

\usepackage{graphicx}%
\usepackage{multirow}%
\usepackage{amsmath,amssymb,amsfonts}%
\usepackage{amsthm}%
\usepackage{mathrsfs}%
\usepackage[title]{appendix}%
\usepackage{xcolor}%
\usepackage{textcomp}%
\usepackage{manyfoot}%
\usepackage{booktabs}%
\usepackage{algorithm}%
\usepackage{algorithmicx}%
\usepackage{algpseudocode}%
\usepackage{listings}%
\usepackage{tikz}%
\usetikzlibrary{intersections, calc}
\usepackage{pgfplots}
\pgfplotsset{compat=1.15} 

\theoremstyle{thmstyleone}%
\newtheorem{theorem}{Theorem}[section]
\theoremstyle{thmstyletwo}%

\theoremstyle{thmstylethree}%
\newtheorem{definition}{Definition}[section]%
\newtheorem{lemma}[theorem]{Lemma}%
\numberwithin{equation}{section}
\begin{document}

\title[Article Title]{Mathematical epidemiology of infectious diseases: an ongoing challenge}

\author*[1]{\fnm{Odo} \sur{Diekmann}}\email{O.Diekmann@uu.nl}

\author[2]{\fnm{Hisashi} \sur{Inaba}}\email{inaba57@u-gakugei.ac.jp}

\author[3]{\fnm{Horst R.} \sur{Thieme}}\email{hthieme@asu.edu}
\affil*[1]{\orgdiv{Mathematical Institute}, \orgname{Utrecht University}, \orgaddress{\street{P.O.Box 80.010}, \city{Utrecht}, \postcode{NL3508TA}, \country{The Netherlands}}}

\affil[2]{\orgdiv{Department of Education}, \orgname{Tokyo Gakugei University}, \orgaddress{\street{4-1-1 Nukuikita-machi}, \city{Koganei-shi}, \postcode{184-8501}, \state{Tokyo}, \country{Japan}}}

\affil[3]{\orgdiv{School of Mathematical and Statistical Sciences}, \orgname{Arizona State University}, \orgaddress{\city{Tempe}, \postcode{AZ 85287-1804}, \country{USA}}}



\abstract{The aim of this paper is to present an, admittedly somewhat subjective, bird's eye view of the mathematical theory concerning the spread of an infectious disease in a susceptible host population with static structure, culminating in a future-oriented description of various modelling challenges.}




\maketitle

Dedicated to the memory of Professor Masaya Yamaguti (1925-1998), who brought a long time ago two of us together and inspired us then and now to view mathematical abstraction as beneficial for scientific comprehension.

\section{The early days}\label{sec1}

Professor Masaya Yamaguti was a man of many talents and wide interests in mathematics, culture and art.
So was Sir Ronald Ross  (1857-1932) who, in 1902, received the Nobel Prize in Physiology or Medicine for his discovery, by way of experimentation, that the parasite causing malaria was transmitted by mosquitoes.  Ross also wrote poems and novels, he composed songs and he was a keen amateur mathematician. This helped him to realise that his discovery opened the door to combat malaria by reducing the prevalence of mosquitoes. In order to quantitatively elaborate this idea, Ross built mathematical models and analysed them\footnote{\url{https://faculty.washington.edu/smitdave/theory/}} \cite{SmBa}, sometimes together with Hilda Hudson, a mathematician at Cambridge University\footnote{\url{https://faculty.washington.edu/smitdave/theory/math_bibliography.html}} \cite{RoHu1, RoHu2} and the first invited female speaker at an International Congress of Mathematicians (ICM) \cite{MiRo}.

Even before the WHO was founded in 1948, this work triggered the Rockefeller Foundation to try out whether elimination of malaria can be achieved in practice. They chose to focus on the Italian island of Sardinia, which was relatively isolated from the rest of the world and not too large. The project ran from 1946 to 1951 and was ultimately successful, largely due to the massive use of the insecticide DDT. When subsequently trying to apply the same methodology in other parts of the world, problems arose. On a longer time scale, mosquitoes developed resistance to DDT. Even more importantly, it was found that the use of DDT has many adverse environmental effects. So in the 1970s the use of DDT was gradually banned all over the world. Today, malaria is still a most severe public health problem in large parts of the world
\cite{Diop2025, Feng2004, Hoppensteadt2011, QuPo, WHO}\cite[Sec.4.1]{Martcheva2015}, as it has been since ancient times \cite{McNeill1976, Michel2024}.

     In 1901, Ross went on a mission to Sierra Leone, to improve the local anti-malarial measures. Among the other people taking part in the mission was the 25-year-old Anderson Gray McKendrick, who had been trained in military and tropical medicine and who, as it has turned out, had an even greater talent for mathematics than Ross himself\footnote{\url{https://mathshistory.st-andrews.ac.uk/Biographies/McKendrick/} }. In 1927 the paper “A contribution to the mathematical theory of epidemics”, written by W.O.Kermack and A.G. McKendrick, was published in the Proceedings of the Royal Society of London, Series A, Containing papers of a mathematical and physical character \cite{Kermack1927}. In our opinion, this remarkable (and remarkably early) paper is THE highlight of the field! Much to our chagrin, it is often misquoted in the sense that it is suggested that this paper introduced the SIR compartmental model and nothing more \cite{Diekmann2020}. In fact, it introduced a rather general model, which mathematically takes the form of a nonlinear renewal equation!

The organization of the paper is as follows. In Sections 2 and 3 we present an informal account of known results concerning various aspects of the Kermack-McKendrick model, together with relevant references. Section 2 concerns the original `homogeneous' model, while Section 3 is devoted to the variant that includes general static heterogeneity. In both cases we explain how compartmental models are included, how one should look at the model from a dynamical systems point of view, how discrete time (and discrete trait) versions can be used for numerical computations, how the initial phase can be characterized in terms of the Basic Reproduction Number $R_0$ and the Malthusian parameter $r$, the notion of Herd Immunity Threshold, and the final size equation and its probabilistic interpretation. Section 4 is devoted to a precise mathematical formulation of a new threshold result. By focussing on the final size, rather than the behaviour in the initial phase, this result can be proven while assuming very little concerning the (regularity of the) model ingredients \cite{Thieme2024b}. In the final Section 5 we describe why ‘real life’ epidemic outbreaks should motivate us to drop the assumption that the trait describing the heterogeneity is static. As allowing the trait to be dynamic ‘destroys’ the monotonicity inherent in the static model, new methods are needed. To emphasize this, we use the word ‘challenges’ in the title of Section 5

This survey paper is very much guided by our own research interests. For the convenience of the reader, we now provide a (non-exhaustive) list of books about epidemic dynamics: \cite{Allen2015, Anderson1991, Bailey1975, Bjornstad2023, Brauer2019, Brauer2008, Britton2019, Busenberg1993,Capasso1993, Chowell2009, Chowell2016, Daley1999, Dieckmann2002, Diekmann2013, Ducrot2022, Feng2014, Hadeler2017, Hernandez2023, Hoppensteadt2011, Iannelli1995, Iannelli2017, Inaba2017, Keeling2008, Kiss2017, Lauwerier1984, Li2018, Li2020, Ma2009, Ma2009b, Magal2008, Malchow2008, Manfredi2013, Martcheva2015, Murty2022, Rass2003, Sattenspiel2009, Smith2011, Takeuchi2007, Thieme2003, Von2023, Vyn2010, Yan2019, Zhao2017}. 
In these books one finds stimulating discussions of a variety of aspects of the rather dynamic field of mathematical epidemiology.

\section{The 1927 Kermack-McKendrick model}

In this section, we provide a motivated description of the 1927 model and its key properties.

\subsection{Model formulation}

The probability per unit of time with which a susceptible individual gets infected is called the {\it force-of-infection} (in analogy with the force-of-mortality, i.e., the instantaneous per capita death rate, in demography; both are examples of the general notion of hazard rate). A key assumption of the 1927 model is that all individuals experience the same force-of-infection (in the next section we shall lift this restriction). The number of new cases (i.e., susceptible individuals who become infected) per unit of time is called the {\it incidence}.

Let $S$ denote the size of the subpopulation of susceptible individuals, and let $F$ denote the force-of-infection. In a closed population, the incidence is equal to both $F S$ and to $- dS/dt$, so we have
\begin{equation}\label{2.1}
\frac{dS}{dt} = - F S.
\end{equation}

The essence of the Kermack and McKendrick model is the constitutive equation that expresses $F$ in terms of the contributions of individuals who became infected before the current time:
\begin{equation}\label{2.2}
 F(t) = \int_{0}^{\infty}A(\tau) F(t-\tau)S(t-\tau) d\tau.
\end{equation}

Here the one-and-only (apart from $N$, the total host population size) model ingredient $A$ describes the expected contribution to the force-of-infection as a function of the time $\tau$ elapsed since infection took place. In this top-down approach we postpone a specification of the stochastic processes that underlie the word {\it expected} (in general, these concern both within host processes, in particular the struggle between the pathogen and the immune system, and the between host contact process). Indeed, Kermack and McKendrick wanted to know what general conclusions could be drawn {\it without} providing such a specification.

We define the susceptible fraction $s(t)$ by
\begin{equation}
s(t)=\frac{S(t)}{N,}
\end{equation}
and the cumulative force-of-infection $w(t)$ by
\begin{equation}
w(t)=\int_{-\infty}^{t}F(\sigma)d\sigma.
\end{equation}
Integrating \eqref{2.1} with respect to time over $(-\infty,t]$ gives the identity
\begin{equation}\label{st}
s(t)=e^{-w(t)}.
\end{equation} 
Next replace $FS$ by $-S'$ in \eqref{2.2} and integrate, then changing the order of the integrals leads to the Renewal Equation (RE)
\begin{equation}\label{RE}
  w(t) = \int_{0}^{\infty} A(\tau) \Psi(w(t - \tau)) d\tau,
  \end{equation}
where 
\begin{equation}\label{psi}
  \Psi(w) := N ( 1 - e^{-w}),
\end{equation}
which corresponds to the subpopulation of no longer susceptible individuals, given the cumulative force of infection.

\subsection{Compartmental models are included}

Assume that a positive integer \( n \), two positive \( n \)-column vectors \( U \) and \( V \) and a positive-off-diagonal $n \times n$ matrix \( D \) are given. If \( A \) has the special form

\begin{equation}
A(\tau) = U^{{\rm T}} e^{D \tau} V,
\end{equation}
where the row vector $U^{{\rm T}}$ denotes the transpose of the column vector $U$, the renewal equation is a rewritten version of the compartmental model

\begin{equation}
\begin{aligned}
&\frac{dS}{dt} = -FS,\cr
&\frac{dy}{dt}=D y+(FS)V,\cr
&F = U^{{\rm T}} y,
\end{aligned}
\end{equation}
(see Section 9.3 in \cite{Diekmann2018}). For instance, the SEIR model corresponds to \( n=2 \),
\begin{equation}
V = \begin{pmatrix} 1 \\ 0 \end{pmatrix}, \quad U = \begin{pmatrix} 0 \\ \beta \end{pmatrix}, \quad D = \begin{pmatrix} -\nu & 0 \\ \nu & -\alpha \end{pmatrix}.
\end{equation}

The general idea is that \(y \) describes the size and composition of the subpopulation of individuals
that were infected in the past and might contribute to the current and future force of infection (so ``removed'' individuals are no longer included). Such individuals undergo a Markov process with transition matrix \(D\). The vector \(V\) describes the probability distribution of the state of individuals that enter this subpopulation, so of the newly infected individuals. (Often all components of \(V\) will be zero except one that equals one. But when, for instance, one third of the individuals is asymptomatic, there will be two non-zero components, one equal to \(1/3\) and the other to \(2/3\).)

The vector \(U\) describes how much an individual contributes to the force of infection, given its state. Our key point is that everything that we shall observe below concerning the RE \eqref{RE}, immediately yields results for a gigantic collection of compartmental models.

\subsection{The dynamical system perspective}

A renewal equation (RE) is a delay equation. And a delay equation is a rule for extending a function of time towards the future, on the basis of the (assumed to be) known past. For Delay Differential Equations, the rule takes the form of a differential equation for the value in the point of extension. But, for an RE, the rule specifies that value itself, like in the (linear and translation invariant) example:
\begin{equation}
x(t) = \int_{0}^{\infty} A(\tau) x(t - \tau) d\tau.
\end{equation}

With a delay equation we associate a dynamical system by translation along the extended function. So even though time is a one-dimensional variable, it makes notational sense to denote the current time by \( t \) and to introduce a bookkeeping variable \( \theta \) to keep track of how far in the past things happened. In particular, we define
\begin{equation}
x_t(\theta) := x(t + \theta), \quad \theta \le 0,
\end{equation}
and prescribe for the time \( t = 0 \) an initial condition
\begin{equation}
x_0 = \phi,
\end{equation}
where \( \phi \) is a given function of \( \theta \). The rule is now used to construct \( x(t) \) for \( t > 0 \). (Note that $x$ depends on $\phi$, even though we did not express this in the notation). Via
\begin{equation}
T(t) \phi := x_t,
\end{equation}
we define a semigroup of operators, in other words, a dynamical system. We refer to \cite{Diekmann1995,Diekmann2007,Diekmann2008,Diekmann2011} for relevant theory concerning this class of dynamical systems.

A natural initial condition for the Kermack-McKendrick model consists of the combination of \( S(0) \) and the incidence $\phi(\theta)$ at time \( \theta \leq 0 \). One can use \( N = S(0) + \int_{-\infty}^{0} \phi(\theta) d\theta \) to eliminate \( S(0) \).

\subsection{Computational aspects}\label{sec2.4}

Although numerical methods have been developed (see in particular \cite{BCCR, BrRiVe, DeSeVe, Messina2022a, Messina2022b, Messina2023a, Messina2023b, Messina2024}), a user-friendly tool to solve equation \eqref{RE} numerically does not exist. By way of pseudo-spectral approximation one can reduce a nonlinear RE to a system of ODEs, thus creating the possibility of a numerical bifurcation analysis with well-tested tools (see \cite{Scarabel2021}). But this approach is not very attractive if one wants to use \eqref{RE} in a modeling context. We recommend that, instead, one works with the discrete time version of \eqref{RE}, as introduced in \cite{Diekmann2021}, and described below.

There are, in fact, good reasons to work with a discrete time formalism: 
\begin{itemize}
    \item[i)] the day-night rhythm has a strong impact on behaviour;
    \item[ii)] public health administration uses one day as the basic time unit.
\end{itemize}

In line with the interpretation of force-of-infection, we replace \eqref{2.1} by the relation
\begin{equation}\label{dis1}
s(t+1) = e^{-\hat{F}(t)} s(t)
\end{equation}
for the day-to-day change in the susceptible fraction \(s\). Here \(\hat{F}(t)\) is the force-of-infection cumulative over the time interval \([t,t+1]\). Often modelers assume, implicitly or explicitly, that \(\hat{F}(t)\) is so small, since the time interval is small, that one can replace \(e^{-\hat{F}(t)}\) by the first two terms
${1 - \hat{F}(t)}$ of the Taylor expansion. This approximation reduces the computational costs (especially when, in an MCMC approach for parameter estimation, one has to repeat computations again and again), but it destroys the fundamental structure!

To complete the model formulation, we replace \eqref{RE} by
\begin{equation}\label{dis2}
\hat{F}(t) = \sum_{k=1}^{\infty} A_k (1 - e^{-\hat{F}(t-k)}) s(t - k)
\end{equation}
where, in practice, we implement an upper bound on the indices $k$ for which \(A_k\) is positive. The finite set of positive \(A_k\) constitute the essential parameters of the discrete time Kermack-McKendrick model.

In \cite{Diekmann2021},  it is shown that all the results that we are going to derive in the rest of this section for \eqref{RE} do have a natural counterpart for \eqref{dis1}-\eqref{dis2}. For a COVID motivated use of a discrete time model, see \cite{Kreck2022}.

\subsection{The initial phase}

An epidemic outbreak starts with small numbers of infected individuals. By demographic stochasticity, the pathogen can go extinct, even if it has the potential for exponential growth, see Section 1.2.2. in \cite{Diekmann2013}. Here we ignore this very early phase and focus on the situation where there are very many infected individuals who form a \underline{small fraction} of the host population.

Noting that with \( \Psi \) defined by \eqref{psi} we have

\begin{equation}
\Psi'(0) = N,
\end{equation}
we write the linearized version of the renewal equation \eqref{RE} as

\begin{equation}
x(t) = N \int_{0}^{\infty} A(\tau) x(t - \tau) d\tau.
\end{equation}
This is exactly the linear RE considered by Lotka and Feller in the context of population dynamics (see Chapter 10 in \cite{Bacaer2011}, Chapter 1 in \cite{Inaba2017}). The equation is both linear and translation invariant, which makes it natural to look for solutions of the special form
\begin{equation}
x(t) = e^{\lambda t},
\end{equation}
since an exponential function is characterized by the property that a translate is a multiple.
Thus we find the so-called Euler-Lotka characteristic equation
\begin{equation}\label{EL}
1 = N \int_{0}^{\infty} A(\tau) e^{-\lambda \tau} d\tau.
\end{equation}
For given \( N \) and \( A \), we first consider the right-hand side as a function of the real variable \( \lambda \). The value at zero is denoted by \( R_0 \), i.e.,
\begin{equation}\label{R0}
R_0 := N \int_{0}^{\infty} A(\tau) d\tau.
\end{equation}

This is called the Basic Reproduction Number, since it has the interpretation of the expected number of secondary cases per primary case in a susceptible population of size \( N \). The real root \( \lambda=r \) of \eqref{EL} is called the Malthusian parameter (this root exists unless \( A \) has a fat tail and the Laplace transform of \( A \) ceases to exist before the value 1 is reached). In general, $\Re \lambda<r$ for all complex roots (the exception $\lambda=r$ occurs for instance when we replace $A$ by a measure concentrated in one point; see Feller \cite{Feller1966} for precise conditions).  
Also note that ${\rm sign}(R_0-1)={\rm sign}~ r$, cf. Figure 1.

\

\begin{center}
    \textbf{Figure 1}
\end{center}

\begin{center}
\begin{tikzpicture}
    \draw[->] (-0.5,0) -- (5,0) node[right] {$\lambda$};
    \draw[->] (0,-0.5) -- (0,3);
    
    \node[left] at (0,2) {$R_0$};
    
    \draw[thick,domain=-0.5:4,samples=100] plot (
        \x, {2*exp(-0.5*\x)}
    );
    
    \draw[dashed] (0,1) -- ({ln(2)/0.5},1);
    \node[left] at (0,1) {$1$};
    
    \draw[dashed] ({ln(2)/0.5},1) -- ({ln(2)/0.5},0) node[below] {$r$};

\end{tikzpicture}
\end{center}

\

For an emerging disease, we can estimate \( r \) from the observed growth rate of the incidence. What we often want is an estimate of \( R_0 \) in order to determine the control effort (for instance, the reduction of the contact intensity) needed to stop the exponential growth. To translate an estimate of \( r \) into an estimate of \( R_0 \), we need information about the "shape" of the graph of \( A \). In this context, the notions of "generation time" and "serial interval" arise.

\subsection{The herd immunity threshold}

Assume that \( R_0 = N \int_{0}^{\infty} A(\tau) d\tau > 1 \). At time \( t \), the susceptible subpopulation is reduced from \( N \) to \( S(t) \). We define
\begin{equation}
R_{\text{eff}}(t) = S(t) \int_{0}^{\infty} A(\tau) d\tau = s(t) R_0,
\end{equation}
and say that the situation is still supercritical if \( R_{\text{eff}}(t) > 1 \), while it is subcritical if \( R_{\text{eff}}(t) < 1 \). The transition occurs when \( t = t^* \) with \( t^* \) defined by the relation
\begin{equation}
s(t^*) = \frac{1}{R_0},
\end{equation}
and we say that {\it herd immunity} is reached at time \( t^* \), to indicate that, on average, primary cases will from now on produce less than one secondary case. Please note that still there may be many new cases, for the simple reason that at \( t^* \) there is a large reservoir of already infected individuals. For the SIR compartmental model, reaching herd immunity coincides with the prevalence \( I \) and the force-of-infection \( F=\beta I \) reaching their maximum, but this is a special feature of that particular model and not a general fact.
We refer to \cite{Nguyen2023} for an analysis of the overshoot of the peak.
The complementary fraction \( 1 - \frac{1}{R_0} \) of victims-so-far is called the {\it Herd Immunity Threshold}.

\subsection{The final size}

The interpretation makes clear that the solution \( w \) of \eqref{RE}, with an appropriate initial condition, is a monotone non-decreasing function of time. By the boundedness of \( \Psi \) and the assumed integrability of \( A \), \( w \) has to be bounded. So the limit \( w(\infty) \) exists. By passing to the limit in \eqref{RE}, we find that
\begin{equation}
w(\infty) = \int_{0}^{\infty} A(\tau) d\tau \cdot \Psi(w(\infty)),
\end{equation}
which, using \eqref{st}, \eqref{psi}, and \eqref{R0}, we can rewrite as
\begin{equation}\label{sinf}
s(\infty) = e^{-R_0 (1 - s(\infty))}.
\end{equation}

Please note that \( s(\infty) \) is
 fully determined by \( R_0 \). This can be understood by viewing \eqref{sinf} as a probabilistic consistency condition.
On the left-hand side, we have the fraction that escapes infection, while on the right-hand side \( R_0 (1 - s(\infty)) \) is the total force of infection in an outbreak in which a fraction \( 1 - s(\infty) \) gets infected (note that it makes sense to consider the expected contribution, since a fraction of a large host population is a large number). So the right hand side of \eqref{sinf} equals the probability that a susceptible individual escapes infection. In Figure 2, we depict \( 1 - s(\infty) \) as a function of \( R_0 \). Please note the sharp rise beyond the value 1 of \( R_0 \).

\


\begin{center}
    \textbf{Figure 2}
\end{center}

\begin{center}
\begin{tikzpicture}
    \draw[->] (-0.5,0) -- (5,0) node[right] {$R_0$};
    \draw[->] (0,-0.5) -- (0,3);
    
    \node[left] at (0,3) {$1-s(\infty)$};
    
    \draw[thick,domain=1:5,samples=100] plot (
        \x, {2-2*exp(-1.5*(\x-1))}
    );
    
    \draw[dashed] (0,2) -- (5,2);
    \node[left] at (0,2) {$1$};

    \node[below] at (1,0) {1};

       \end{tikzpicture}
\end{center}
   
\

\section{Incorporating static heterogeneity}

Again the total population size will be denoted by $N$, but now we assume that individuals are characterized by a trait $x$, taking values in a measurable space $\Omega$, so in a space equipped with a $\sigma$-algebra. The trait is static, i.e., it does not change in the course of time. The composition of the population is described by the unit measure $\Phi$ on $\Omega$. The use of measures allows us to combine continuum models, with the trait distribution described by a density, and discrete models, involving countably many types of individuals, in one and the same framework.  

We refer to \cite{Bootsma2024, Diekmann2023, Inaba2025, Thieme2024b} for recent relevant theory concerning the Kermack-McKendrick model with static heterogeneity. The first trait to be considered may have been spatial location \cite{Diekmann1978, Thieme1977}, other traits have been susceptibility \cite{Gomes2022}, resistance to infection \cite{Ponce2024}, mask compliance \cite{Bootsma2023} and household size \cite{Luckhaus2023a, Luckhaus2023b}.  

Please note that, in the case of spatial position and non-compact domains, outbreak dynamics is dominated by spatial expansion with an asymptotic speed of propagation corresponding to the lowest possible speed of travelling wave solutions. Here we do not discuss this interesting phenomenon, but refer interested readers to \cite{Alanazi2020}, \cite{Rass2003}, \cite{Roq2024} and the references given there.

\subsection{Model formulation}

Apart from $N$ and $\Phi$, we just need one model ingredient:

\begin{equation}\label{A}
\begin{aligned}
A(\tau,x,\xi)=&\text{the expected contribution to the force of infection on an individual}\cr 
&\text{with trait $x$ of an individual with trait \(\xi\) that became infected \(\tau\) }\cr
&\text{units of time ago}
\end{aligned}
\end{equation}

Let $S(t, \omega)$ denote the number of susceptibles, at time $t$, with trait belonging to the (measurable) subset $\omega$. We assume:

\begin{equation} \label{eq:3.2}
    S(t, \omega) = N \int_\omega s(t, x) \, \Phi(dx),
\end{equation}
where, for each $t$, $s(t, \cdot)$ is a measurable function defined on $\Omega$ and taking values in $[0, 1]$. In words: $s(t, x)$ is the trait-specific fraction that is still susceptible at time $t$.

With $F(t, x)$ denoting the force-of-infection at time $t$ on individuals with trait $x$, and assuming that $s$ is partially differentiable with respect to $t$, we have:

\begin{equation} \label{3.3}
    \frac{\partial s}{\partial t}(t, x) = -F(t, x) s(t, x),
\end{equation}
and, by integration,

\begin{equation} \label{3.4}
    s(t, x) = \exp\left(-\int_{-\infty}^t F(\tau, x) \, d\tau\right).
\end{equation}

The constitutive equation for $F$ now takes the form:

\begin{equation} \label{3F}
    F(t, x) = \int_{0}^{\infty}\int_\Omega A(\tau,x, \xi) \left(-N\frac{\partial s}{\partial t}(t-\tau, \xi)\right) \Phi(d\xi) d\tau,
\end{equation}
and when we integrate both sides of this equality with respect to time from $-\infty$ to $t$, interchange the order of the integrals and substitute the result into \eqref{3.4}, we obtain the abstract renewal equation (RE):

\begin{equation} \label{3RE}
s(t,x) = \exp\left(-N\int_{0}^{\infty}\int_{\Omega}A(\tau,x,\xi)[1-s(t-\tau,\xi)]\Phi(d\xi)d\tau\right)
\end{equation}
which serves as the starting point for the analysis. In the present section, we provide information about the current knowledge concerning \eqref{3RE}, while following exactly the same scheme as in Section 2. In Section 4, we shall formulate and prove some new results.

\subsection{Compartmental models are included}

When functions $a$, $b$ and $c$ exist such that
\begin{equation}\label{abc}
   A(\tau, x, \xi) = a(x) b(\tau) c(\xi),
   \end{equation}
we deduce from \eqref{3RE} that 
\begin{equation}\label{aw}
s(t,x)=e^{-a(x)w(t)}.
\end{equation}
with
\begin{equation}\label{w}
w(t)=N\int_{0}^{\infty}\int_{\Omega}b(\tau)[1-s(t-\tau,\xi)]c(\xi)\Phi(d\xi)d\tau.
\end{equation}
or, equivalently,
\begin{equation}\label{3.9} 
 w(t) = \int_{0}^{\infty} b(\tau) \Psi(w(t-\tau)) d\tau,
 \end{equation}
with $\Psi$ now defined by
\begin{equation}\label{3.10}
\Psi(w) := N \int_{\Omega}  c(\xi) (1 - e^{- a(\xi) w} ) \Phi(d\xi). 
\end{equation}

 Note that from \eqref{3.10} we recover the earlier definition \eqref{psi} in the homogeneous case that both $a$ and $c$ are identically equal to 1.
 Also note that when $a$ is identically equal to one, \eqref{3.10} states that we can simply work with the average value of $c$. 

    Essentially, \eqref{RE} and \eqref{3.9} are the same. When $a$ is not constant, \eqref{3.10} differs from \eqref{psi}.
   When
\begin{equation}\label{3.11}
  	b(\tau) = U^{{\rm T}} e^{\tau D} V,
  	\end{equation}
we can define a vector $Q$ as 
\begin{equation}\label{3.12}
 \begin{aligned}
  Q(t) :&= \int_{0}^{\infty} e^{\tau D} V \Psi(w(t-\tau)) d\tau \cr
  &= \int_{-\infty}^{t} e^{(t-\sigma)D} V \Psi(w(\sigma)) d\sigma.
\end{aligned}
\end{equation}
Then it follows that
\begin{equation}\label{3.13}
\frac{dQ}{dt} = D  Q + V \Psi(w).
\end{equation}
On the other hand, it follows from \eqref{3.11}, \eqref{3.12} and \eqref{RE} that

\begin{equation}\label{3.14}
  w(t) = U^{{\rm T}} Q(t).
\end{equation}
 Combining \eqref{3.13} and \eqref{3.14} we obtain the closed system of ODE (the {\it integrated form} of the compartmental model, see \cite{Diekmann2023}):
\begin{equation}\label{3.15}
\frac{dQ}{dt} = D  Q + V \Psi(U^{{\rm T}} Q).
\end{equation}

Note that, conversely, given a solution of \eqref{3.15}, we can define $w$ by \eqref{3.14} and verify that $w$ satisfies \eqref{3.9}.

We conclude that in terms of the integrated formulation of a compartmental model, we can incorporate separable heterogeneity by simply redefining the function $\Psi$. The fact that \eqref{3.10} involves an integral is of little to no importance for the theory. But when one wants to study \eqref{3.15} numerically, it is a nuisance. In special cases, one can replace the integral by an explicit expression. 

\

{\bf Example:}
Let $\Omega = [0, \infty)$ and let $a(x) = x$, i.e., let the trait correspond directly to relative susceptibility. 
Let $\Phi$ be the Gamma Distribution with mean 1 and variance $p^{-1}$ . In other words, let $\Phi$ have density
\begin{equation}\label{g7.1}
       x \to   \frac{p^p}{\Gamma(p)}  x^{p-1}  e^{-px}.
       \end{equation}

The key feature is that under these assumptions we can evaluate the integral in \eqref{3.10} when $c$ is a (low order) polynomial and thus obtain an explicit expression for $\Psi$. The underlying reason is that we deal with (a derivative of) the Laplace Transform of $\Phi$, which is itself explicitly  given by
\begin{equation}\label{g7.2}
   \hat{\Phi}(\lambda)  = \left(\frac{\lambda}{p} + 1\right)^{-p}.
\end{equation}

If the trait has no influence on infectiousness, i.e., $c$ is identically equal to 1, we have
\begin{equation}\label{g7.3}
   \Psi(w) =N\left[ 1 - \left(\frac{w}{p} + 1\right)^{-p} \right],
\end{equation}
while if infectiousness too is proportional to the trait, i.e., $c(\xi)=\xi$, we obtain
\begin{equation}\label{g7.4}
   \Psi(w) = N\left[1 - \left(\frac{w}{p} + 1\right)^{-p-1}\right].
   \end{equation}
In \cite{Bootsma2024}, we compare and contrast these special cases in terms of $R_0$, the Herd Immunity Threshold and the final size.

\subsection{The dynamical system perspective}

Equation \eqref{3RE} is an abstract RE. This name indicates that the unknown $s$ is a function of time, taking values in an infinite-dimensional space, viz., the space of bounded measurable functions on $\Omega$. For general theory, we refer to \cite{Janssens2020} and the references given there.
For equation \eqref{3RE}, we can actually employ monotone iteration to constructively define its solution. This will be done in Section 4.

\subsection{Computational aspects}

As far as we know, sophisticated numerical methods that directly apply to \eqref{3RE} do not exist. Moreover, it seems unlikely that a modeler has enough information to identify the measure $\Phi$ and the function $A$ of three variables. Pragmatism therefore suggests working with discrete time steps, cf. (2.15)-(2.16), and to restrict to finitely many types, i.e., a measure $\Phi$ supported in finitely many points. In this case, the model ingredient $A$ is replaced by a family of (contact) matrices indexed by the discrete 'age of infection,' and numerical simulations are easy to perform (which certainly helps to make MCMC parameter identification feasible).

\subsection{The initial phase}

Equation \eqref{3RE} admits the disease-free steady-state solution $s=1$ identically. 
 Inserting
\begin{equation}
s(t,x):=1-y(t,x),
\end{equation}
into \eqref{3RE} and assuming that \(y\) is small, we obtain, upon neglecting the higher order terms in the Taylor expansion, the linearized equation
\begin{equation}\label{y}
y(t,x)=N\int_{0}^{\infty}\int_{\Omega}A(\tau,x,\xi)y(t-\tau,\xi)\Phi(d\xi)d\tau.
\end{equation}

We refer to Sections 5 and 6 of \cite{Thieme1985} for an early profound analysis of such linear equations within the setting of positive operator theory. The general idea is to define the Basic Reproduction Number $R_0$ in two steps \cite{Diekmann1990}: first, one defines the Next Generation Operator (NGO) by
\begin{equation}\label{R0}
(K\phi)(x):= N \int_{\Omega} k(x,\xi)\phi(\xi) \Phi(d\xi),
\end{equation}
with
\begin{equation}\label{k}
k(x,\xi): = \int_{0}^{\infty} A(\tau,x,\xi)d\tau,
\end{equation}
and then one defines $R_0$ as the spectral radius of the NGO: 
\begin{definition}\label{defR0}
$R_0 :=\rho(K)$,
where $\rho(K)$ denotes the spectral radius of the operator $K$.
\end{definition}

The aim is to demonstrate that $R_0$ has a threshold value of 1, meaning in particular that $y$ grows exponentially when $R_0 > 1$ while decaying exponentially when $R_0 < 1$. In this context, it is helpful when $R_0$ is actually an eigenvalue of the NGO, so one makes assumptions on the model ingredients that guarantee that this is indeed the case. We refer to \cite{Inaba2025} for a recent account of this approach. Also see \cite{Franco2023}.

In Section 4, we shall formulate and prove some new threshold results under much weaker assumptions. These are formulated in terms of the final size. The proof exploits a link, established in \cite{Bootsma2024}, between the NGO in terms of fractions and the final size equation.

\subsection{The herd immunity threshold}

Linearising the RE \eqref{3RE} at time $t$, we obtain
  \begin{equation}\label{z}
z(t,x)=Ns(t,x)\int_{0}^{\infty}\int_{\Omega}A(\tau,x,\xi)z(t-\tau,\xi)\Phi(d\xi)d\tau,
\end{equation}
where $z(t,x)$ is a small perturbation from $s(t,x)$.  For a fixed time $t$, we can define the {\it effective next generation operator} at time $t$ as
\begin{equation}\label{Re}
(K_{\text{eff}}\phi)(x):= N s(t,x)\int_{\Omega} k(x,\xi)\phi(\xi) \Phi(d\xi).
\end{equation}

If we replace, at the right hand side of \eqref{Re}, $s(t,x)$ by the vaccination coverage, we can similarly define the effective next generation operator for the vaccination. 
Then the {\it effective reproduction number} $R_{\text{eff}}$ is defined as the spectral radius of $K_{\text{eff}}$ \cite{Inaba2025}.

The condition $R_{\text{eff}} = 1$ defines a codimension-one manifold in an infinite-dimensional space. It depends on the initial condition (so on the precise way in which the outbreak is triggered) where the orbit will intersect this manifold. As a consequence, the HIT is not a well-defined general concept.

But in the special case that the kernel $A$ factorizes, one can revert back to the one-dimensional situation, define the HIT, and analyze how it is influenced by the heterogeneity. This is done in Sections 6 and 7 of \cite{Bootsma2024}, and it leads to epidemiologically relevant insights.

\subsection{The final size equation}

Exactly as the interpretation suggests, $t \mapsto s(t,x)$ is monotone non-increasing for each $x$ in $\Omega$. In Section 4, we show that, by passing to the limit, one obtains from \eqref{3RE} a final size equation involving the NGO in terms of fractions and a substitution/Nemytskii operator. For proving the threshold property, this characterization of the final size is key.

\section{Threshold theorem}

 Recall the interpretation of $A(\tau,x,\xi)$ presented in (3.1). So far we did not yet explicitly formulate relevant  assumptions concerning $A$, but we shall do so now. We assume  that the function
\[           A : \mathbb R_+ \times \Omega \times \Omega \to  \mathbb R \]
\begin{enumerate}
\item[]— is non-negative and measurable, 
\item[]— is integrable over $\mathbb R_+ \times \Omega$ with respect to the first and third variables,
\item[]— is such that k defined by (3.24) is bounded.
\end{enumerate}

We now first explain how to construct the relevant solution of \eqref{3RE}. Here we essentially follow \cite{Thieme2024b}. For $t > 0$, we replace \eqref{3RE} by
\begin{equation}\label{4RE1}
s(t,x) = \exp\left(-N\int_{0}^{t}\int_{\Omega}A(\tau,x,\xi)[1-s(t-\tau,\xi)]\Phi(d\xi)d\tau-h(t,x)\right), 
\end{equation}
\begin{equation}\label{4RE2}
h(t,x) := \int_{0}^{t}F_{0}(\sigma, x)d\sigma.
\end{equation}
where $F_0$ is bounded, measurable, non-negative and integrable over $\mathbb R_+$ with respect to its first argument.

This means that we split the cumulative force-of-infection into the part contributed by individuals that became infected after \( t=0 \) and the part $h$ contributed by individuals that became infected before \( t=0 \). We have in mind that the latter part is rather small. Note that \( s(0,x)=1 \), which is strictly speaking inconsistent with ``individuals infected before \( t=0 \)." The inconsistency is removed if we re-interpret \(\Phi\) and \( N \) as pertaining to the susceptible population at \( t=0 \). The point of ``rather small" is that the modified \(\Phi\) and \( N \) are very close to the overall \(\Phi\) and \( N \).

We construct the solution \(s(t,x)\) of \eqref{4RE1} by monotone iteration, as follows: Define

\begin{equation}
 s_0(t,x) = e^{-h(t,x)},
\end{equation}
and inductively, for \(n \geq 0\),

\begin{equation}\label{sn}
s_{n+1}(t,x) = \exp\left(-N\int_{0}^{t} \int_{\Omega} A(\tau,x,\xi)[1 - s_n(t-\tau,\xi)]\Phi(d\xi)d\tau-h(t,x)\right).
\end{equation}

The monotonicity of \(x \mapsto e^{-x}\) guarantees that $s_1 \le s_0$.  Since the right hand side of \eqref{sn} is non-decreasing with respect to $s_n$, we have inductively 
\(s_{n+1}(t,x) \leq s_n(t,x)\), for all $t > 0$ and \(x \in \Omega\). Clearly \(s_n(t,x) \geq 0\). 
By the Beppo Levi Theorem on monotone convergence, we can pass to the limit at the right hand side of the recurrence relation \eqref{sn}.  We conclude that the limit
\begin{equation}
s(t,x) = \lim_{n \to \infty} s_n(t,x),
\end{equation}
exists and satisfies \eqref{4RE1}.

The ``generation interpretation" of the recurrence relation implies that we constructed the biologically relevant solution of \eqref{4RE1}. So there is no need to discuss uniqueness. But one can anyhow remark that any other solution \(\tilde{s}(t,x)\) with \(0 \leq  \tilde{s}(t,x) \leq 1\) satisfies
\begin{equation}
\tilde{s}(t,x) \leq s(t,x),
\end{equation}
since \(\tilde{s}(t,x) \leq s_0(t,x)\) and therefore \(\tilde{s}(t,x) \leq s_1(t,x)\) and, by induction, \(\tilde{s}(t,x) \leq s_n(t,x)\) for all \(n\). So \(s(t,x)\) is the maximal solution.

Since, by assumption, \(F_0(t,x) \geq 0\) in \eqref{4RE2} the function \(t \mapsto s_0(t,x)\) is monotone non-increasing for all \(x \in \Omega\). By induction the same holds for \(t \mapsto s_n(t,x)\) and therefore the limit \(s(t,x)\) inherits this property. Since \(s(t,x) \geq 0\), the limit \(s(\infty, x)\) exists for all \(x \in \Omega\). Using Beppo Levi's monotone convergence theorem once more, we find that
\begin{equation}\label{4.7}
s(\infty, x) = e^{-\int_{0}^{\infty} F_0(\sigma, x) d\sigma} \cdot e^{-(K(1-s(\infty, \cdot)))(x)},
\end{equation}
where $K$ is the linear operator defined in \eqref{R0}.

 Equation \eqref{4.7} is in fact a probabilistic consistency condition. At the left hand side we have the fraction that escapes infection, while at the right hand side we have the probability to escape infection, given the cumulative force-of-infection generated by individuals infected before and after $t=0$. 

     Note that the second factor at the right hand side of \eqref{4.7} is the composition of a Nemytskii operator and the linear operator $K$ acting on the function $x  \to 1 - s(\infty,x)$. We consider $K$ as a bounded linear operator on the Banach space $M^b(\Omega)$ of bounded measurable functions defined on $\Omega$, equipped with the supremum norm. Since $k$ defined by \eqref{k} is non-negative, $K$ is a positive operator, i.e., $K$ maps the positive cone $M_+^b(\Omega)$, consisting of functions taking non-negative values, into itself. 

     We now explain the interpretation of $K$ as the Next-Generation-Operator acting on fractions. Let the function $\phi$ describe, in terms of fractions, the composition of a generation of infected individuals. Provided the generations are, in terms of fractions, very small, the function $K\phi$ describes, again in terms of fractions, the composition of the next generation, i.e., the individuals infected by those described by $\phi$. The ‘very small’ is used when we linearize, i.e., replace “$1 - \exp(- \text{cumulative force-of-infection})$” by “cumulative force-of-infection”. That the next generation is obtained by applying the linear operator $K$ is an approximation!

     Now recall Definition \ref{defR0} : the Basic Reproduction Number $R_0$ is, by definition, the spectral radius of $K$. As already noted in Section 3.5, it is perfectly possible to make additional assumptions on $A$ such that one can prove that
\begin{enumerate}
\item[]— $R_0$ is an eigenvalue of $K$,  

\item[]— the linearized equation (3.22) has a solution of the special form $y(t,x) = e^{rt} \psi(x)$ with $r \in \mathbb R$  and ${\rm sign}(r) = {\rm sign}(R_0 - 1)$,
\item[]— the asymptotic large time behaviour of general positive solutions of \eqref{y} is of this special form,
\end{enumerate}
see \cite{Franco2023, Inaba2025}. But here we follow the approach of \cite{Thieme2024b}, which consists of checking that the adjoint of $K$ has $R_0$ as an eigenvalue and next using the corresponding eigenfunctional to deduce certain estimates for the solution of the final size equation. As we shall see, these estimates do not require additional assumptions and they imply in a somewhat different, yet quite meaningful, way that $R_0$ has threshold value one.

The following two properties
\begin{enumerate}
\item[]— the cone $M_+^b(\Omega)$ has non-empty interior,
\item[]— the supremum norm is monotone,
\end{enumerate}
are important, since they allow us to refer to the rather general Corollary 11.17 in \cite{Thieme2024a} (see also Appendix 2.6 in \cite{Sch1966} and Theorem 3.1 in \cite{Thieme2024b}) as justification of

\begin{theorem}\label{th4.1}
Assume that $R_0 > 0$. Then there exists a positive bounded linear functional $\theta : M^b(\Omega) \to \mathbb R$ such that $\theta(K\phi)=R_0\theta(\phi)$ or, in words, $R_0$ is an eigenvalue of the adjoint of $K$ with eigenfunctional $\theta$.
\end{theorem}

Now let’s return to equation \eqref{4.7}.
In terms of the final cumulative force-of-infection
\begin{equation}
u(x) := - \ln s(\infty, x),
\end{equation}
we can write this equation in the form
\begin{equation}\label{u}
u = h_\infty + F(u),
\end{equation}
where
\begin{equation}\label{h}
h_\infty(x) := h(\infty,x)=\int_{0}^{\infty} F_0(\sigma, x) d\sigma,
\end{equation}
and
\begin{equation}\label{F}
F(u) = K (1-e^{-u(\cdot)}).
\end{equation}

We mention that these equations can alternatively
be derived directly in terms of the cumulative
force of infection \cite{Thieme2024b}.
We shall use equation \eqref{u} to deduce certain properties of any positive solution \( u \). So there is no need to discuss the uniqueness of a positive solution. In fact, in the present generality, there can be many solutions.

Note that, since \( K \) is positive and \( f'(y) \geq 0 \), \( F \) is order preserving. 
 Moreover, we have
 
\begin{lemma}\label{lem4.2}
For all \(\phi \in M^b_+(\Omega)\)
\begin{equation}
e^{-\sup \phi} K\phi \leq F(\phi) \leq K\phi.
\end{equation}
\end{lemma}
\begin{proof} 
Let $f(y)=1-e^{-y}$. The inequality
\begin{equation}
F(\phi) \leq K\phi,
\end{equation}
follows directly from
\begin{equation}
f(y)  \leq y, \, y \geq 0.
\end{equation}
For the other inequality, we use
\begin{equation}
f(y)  \geq ye^{-y},
\end{equation}
which follows by observing that we have equality for \( y = 0 \) while
\begin{equation}
\frac{d}{dy} \left[ 1 - e^{-y} - ye^{-y} \right] = ye^{-y} > 0,
\end{equation}
for \( y > 0 \).
All that remains is to observe that
\begin{equation}
e^{-\phi(x)} \geq e^{-\sup \phi},
\end{equation}
for \(\phi \in M^b_+(\Omega)\).
\end{proof}

For the above Lemma \ref{lem4.2}, the reader may refer to Section 2.2 in \cite{Thieme2024b}.

\begin{lemma}\label{lem4.3}
If \( R_0 < 1 \) then
\begin{equation}
u \leq (I-K)^{-1}h_\infty,
\end{equation}
and consequently
\begin{equation}
s(\infty, x) = e^{-u(x)} \uparrow 1 \quad \text{if} \quad \sup h_\infty \downarrow 0.
\end{equation}
\end{lemma}
\begin{proof}
Let $u$ be a positive solution of \eqref{u}.  It follows from \eqref{F} that
\begin{equation}
u \leq h_\infty + Ku. 
\end{equation}
Then it holds that $u \le (I-K)^{-1}h_\infty$, because $I-K$ is positively invertible if $R_0=r(K)<1$, which establishes the conclusion.
\end{proof}

For the above Lemma \ref{lem4.3}, the reader may refer to Theorem 3.2 in \cite{Thieme2024b} and Proposition 12 in \cite{Inaba2025}.

\begin{lemma}\label{lem4.4}
Let $u$ be a positive solution of \eqref{u}. If \( R_0 > 1 \) and \( \theta(h_\infty) > 0\), then it holds that
\renewcommand{\labelenumi}{(\theenumi)}
\begin{enumerate}
\item \( \theta(u) >  \theta(h_\infty) > 0\),  
\item \(\sup(u - h_\infty) \geq \ln R_0\),
\item \(\inf s(\infty, \cdot) = \inf e^{-u(\cdot)} \leq \frac{1}{R_0}\).
\end{enumerate}
\end{lemma}
\begin{proof}
(1) Since \( F \) is order preserving, \eqref{u} implies $u \geq h_\infty$, so we have
\begin{equation}
  \theta(u) \geq  \theta(h_\infty) > 0.
\end{equation}
Next observe that
\begin{equation}
 \theta(u-h_\infty) =  \theta( F(u)) \geq e^{-\sup u}  \theta( K u) = e^{-\sup u} R_0  \theta( u) > 0.
\end{equation}
Hence
\begin{equation}
\theta( u )>  \theta( h_\infty)> 0.
\end{equation}
(2) Write \eqref{u} as
\begin{equation}
u - h_\infty = F(u),
\end{equation}
and note that
\begin{equation}
F(u) \geq F(u-h_\infty),
\end{equation}
since \( F \) is order-preserving.  It follows from Lemma \ref{lem4.2} that
\begin{equation}
F(u-h_\infty) \geq e^{-\sup(u-h_\infty)} K(u-h_\infty).
\end{equation}
Thus we find
\begin{equation}
u-h_\infty \geq e^{-\sup(u-h_\infty)} K(u-h_\infty).
\end{equation}
Acting on this inequality with \(\theta\) we obtain
\begin{equation}
 \theta( u-h_\infty) \geq e^{-\sup(u-h_\infty)} R_0  \theta( u-h_\infty),
\end{equation}
which can only be true if the factor \( e^{-\sup(u-h_\infty)} R_0 \) is less than or equal to one, since $ \theta( u-h_\infty)>0$. It follows that
\begin{equation}
\sup(u-h_\infty) \geq \ln R_0.
\end{equation}
(3) Since \( e^{-u} \leq e^{-(u-h_\infty)} \) we have
\begin{equation}
\inf e^{-u} \leq e^{-\sup(u-h_\infty)} \leq \frac{1}{R_0}.
\end{equation}
\end{proof}

For the above Lemma \ref{lem4.4}, the reader may refer to Theorem 3.2 in \cite{Thieme2024b}.
The interpretation of (3) of Lemma \ref{lem4.4} is that even a very small introduction has, provided $ \theta( h_\infty)  > 0$, a large effect for at least some traits \( x \).
From Theorem 3.1 in \cite{Thieme2024b}, the condition $ \theta(h_\infty)>0$ is satisfied if $h_\infty \in M^b_+(\Omega)$ and $h_\infty \ge \delta v$ for some $\delta>0$ and $v \in M^b_+(\Omega)$ given by
\begin{equation}
v(x):=  \int_{\Omega} k(x,\xi) \Phi(d\xi).
\end{equation}

Now we have a threshold theorem:

\begin{theorem}
 \( R_0 = 1 \) is a threshold in the sense that for \( R_0 < 1 \) a small introduction can only lead to a small outbreak, while for \( R_0 > 1 \) a small introduction can lead to a large outbreak.
\end{theorem}

Note that our setting includes the case of isolated subpopulations, i.e., \( A(\tau, x, \xi) = 0 \) if \( x \neq \xi \). For this case we can define the basic reproduction number at $x \in \Omega$, denoted by  \( R_0(x) \), and it is possible that \( \{x : R_0(x) \geq 1\} \) consists of just a few points. If so, \(\theta\) corresponds to the evaluation in one or more of these points.
Then it is possible that, when we start with \( R_0 > 1 \), after the outbreak we still have a (possibly different) \( R_0 > 1 \). Indeed, consider two isolated subpopulations, both with \( R_0 > 1 \). The spectral radius corresponds to the largest of the two. When we introduce the infection only in one subpopulation, nothing will change in the other subpopulation. The point is that the NGO may have several positive eigenfunctionals if we deal with a reducible situation.
On the other hand, for the irreducible situation, the effective reproduction number of the susceptible population remaining after the epidemic outbreak (the final reproduction number) is less than one (Theorem 5.6 in \cite{Bootsma2024}, Remark 3.4 in \cite{Thieme2024b}, Proposition 15 in \cite{Inaba2025}).

\section{Challenges / Open Problems}

 With the Covid pandemic fresh in our memory, it is clear that the static heterogeneity Kermack-McKendrick model lacks many desirable properties. Human behaviour adapts to perceived risk, either spontaneously or on the instruction of the government \cite{dOnofrio2021}. Immunity wanes \cite{Scarabel2025}. Vectors adapt their behaviour \cite{Diop2025}. Pathogen strains mutate. Seasonal changes have an impact on the transmissibility of the virus \cite{dOnofrio2024b}. Vaccination may take place during an outbreak \cite{dOnofrio2024a}. 

     Even though it is clear that outbreak dynamics is influenced by all of these mechanisms, it is virtually impossible to quantitatively disentangle their effect on the basis of data. Public health organizations often advise their government by way of scenario analyses performed with large scale simulation models. In our opinion, this should not discourage (applied) mathematicians. Hopefully we can gradually develop and analyse theoretical models that provide general qualitative insights in the influence of the kind of mechanisms listed above. Such insights serve as simplified symbolic ‘pictures’ that can greatly catalyse the meaningful interpretation of numbers and graphs produced by way of computational studies of complex models. Also in the world of models, diversity is desirable.

     As a rule, transmission is superimposed on a contact process. Therefore epidemic modeling often requires modeling of an underlying contact process. As a concrete example we mention Heesterbeek and Metz \cite{Heesterbeek1993} saturating contact process based on pair formation, even though, in a suitable limit, pairs dissolve quickly, i.e., stay together for only a negligible (at the time scale of epidemic dynamics) amount of time. The  original motivation of this contact submodel was to investigate how transmission is, possibly, influenced by population size/density. But it might also serve to incorporate a somewhat mechanistic description of a behavioural feedback loop, as follows. The process is characterized by one parameter, see Section 8 of \cite{Bootsma2024}. By allowing this parameter to vary dynamically in response (possibly with some delay) to the prevailing incidence or prevalence, one obtains a simple representation of contact reduction triggered by awareness. The prediction would be that the overshoot is reduced \cite{Nguyen2023}. Note that dynamic heterogeneity by itself, i.e., without feedback, can also lead to overshoot reduction \cite{Tkachenko2021b}. See also Appendix B1 of \cite{Bootsma2024}.

      We refer to \cite{Heesterbeek1993} and Section 9 of \cite{Bootsma2024} for a heterogeneous version of the Heesterbeek-Metz model and to \cite{Thieme2000} for a general survey. It would be an interesting challenge to incorporate behavioural feedback in such more complicated models. See Appendix B2 of \cite{Bootsma2024} for references to various papers about behavioural feedback. Also see \cite{Buonomo2025} for a recent survey, including both compartmental models and models in Kermack-McKendrick spirit.
      
      In this paper we focused on outbreak models, but malaria and many other infectious diseases are in fact endemic. The demography of the host is an important constituent of models for endemic diseases \cite{Breda2012, Inaba2001, Inaba2016, Kermack1932, Kermack1933}. Far too often, especially in compartmental models, one assumes that host survival is described by an exponential function of age $a$. For modern human society this is way off, even in underdeveloped countries. In our opinion, it is not that difficult to formulate models involving both a demographic survival function and general survival functions that allow us to describe the expected infectiousness as a function of time-since-infection or the expected susceptibility as a function of time-since-recovery, but it is close to impossible to obtain epidemiologically relevant insights by analysing such models (The situation is very different for age-structured epidemic models without duration dependence. see \cite{Anderson1991, Busenberg1993, Dietz1975, Hethcote2000, Iannelli1995, Iannelli2017, Inaba1990, Inaba2017, Li2020, Thieme2003}.)
    For vector transmission (like in the case of malaria or dengue) too, the formulation of models with general infectiousness as a function of time-since-infection is feasible, but the analytic investigation  of such models is a major challenge. (When we say `analysis', we have purely theoretical analysis in mind. Concerning numerical, in particular numerical bifurcation, approaches, we are a bit more optimistic. see \cite{Diekmann2025}.) We refer to the pioneering work of Dietz and Schenzle \cite{Dietz1985}, Gripenberg \cite{Gripenberg1983}, Hoppensteadt \cite{Hoppensteadt1975} and the more recent attempts \cite{Inaba2008, Inaba2016, Breda2012, Burie2017}\cite[Secs. 6.5 and 8.4]{Inaba2017}, for some evidence, but stress that we hope, of course, that our readers will prove us wrong concerning the `impossibility' half of the statement.

A major challenge, in our opinion, is to investigate whether graph limits\footnote{\url{https://lovasz.web.elte.hu//bookxx/hombook-almost.final.pdf}} can yield manageable models that do generate epidemiological insights? Here the starting point would consist of network models in the spirit of \cite{Cure2025, Kiss2017, Leung2017, Zhao2025}.

As for sexually transmitted diseases \cite{Anderson1991, Brauer2019, Inaba2017}, pairwise interactions make perfect sense. But for airborne transmission it also makes sense to think in terms of meeting places and meeting events in which several individuals are involved. Despite the recognized importance of superspreading events \cite{Lloyd2005} and the established tradition of household models, \cite{Ball2002} and the references in there, it seems fair to say that such models are presently making a hesitant start, see \cite{Luckhaus2023a, Luckhaus2023b, Ball2022, Ball2024}. A side effect of larger meeting events may be an enhancement of the dose received by an up-to-then susceptible individual. How dose and effect are related is, in view of all the intricacies of innate immunity, a subtle problem. And when do we speak of one dose and when do we consider the time interval between two doses so large that we consider their effect as independent of each other?  See \cite{Luckhaus2023a} and \cite{Ponce2024} for  models requiring a sufficiently large instantaneous or cumulative dose, respectively, in order to produce a true infection, leading to an Allee effect at the population level.

The mathematical approach in this paper relies on the bilinear structure of the incidence (2.1). A generalization to a nonlinear dependence of $\frac{dS}{dt} $ on $S$ is possible \cite{Kiss2017} and Section 2.5 of \cite{Leung2017}; the function $\Psi$ in (2.7) becomes more complicated but keeps its properties. A direct generalization to a nonlinear dependence on $I$ (as it occurs, e.g., if there is an underlying contact process as described before) does not seem possible.
One can, however,  derive some invariants
which provide a possible explanation that some
covid-19 data  suggest some sort of Allee effect
for the size of the outbreak \cite{Luckhaus2023a, Luckhaus2023b}.

Since monotonicity arguments play a large role in the
derivation of our results, appropriate generalizations
can be obtained for incidences that can be minorized
and/or majorized by density-dependent incidences
(called upper and lower density-dependent incidences in \cite{FaTh}).
In the mathematical derivation, equalities would be replaced
by inequalities. This way, in a model
with infection-age, but without static heterogeneity
(though an extension should be possible), one still obtains
that, after an epidemic outbreak, there always
are some susceptibles left \cite[Thm.4.1]{FaTh}, but that their number
can be very small if some appropriate version of
a basic reproduction number is very large
\cite[Thm.5.2]{FaTh}.

\

     Our conclusion: challenges galore …

\

\bmhead{Acknowledgements}
 We would like to thank the two reviewers for their helpful comments, which improved the manuscript.

\section*{Declarations}

\begin{itemize}
\item Funding: Hisashi Inaba was supported by JSPS KAKENHI Grant Number 22K03433 and Japan Agency for Medical Research and Development (AMED; JP23fk0108685). 

\item Conflict of interest/Competing interests: The authors have no relevant financial or non-financial interests to disclose.
\end{itemize}

\bigskip



\end{document}